\documentclass[a4paper,12pt]{article}
\usepackage{graphicx} 
\usepackage{amsmath, amssymb, amsthm}
\usepackage{nccmath}
\usepackage{fullpage}
\usepackage{natbib}
\usepackage{titlesec}
\usepackage{listings}
\usepackage{xcolor}
\usepackage[toc]{appendix}
\usepackage{multirow}
\usepackage{comment}
\usepackage{subcaption}
\usepackage{siunitx}
\usepackage{hyperref} 
\hypersetup{colorlinks = true, citecolor = black, urlcolor = cyan, linkcolor = black}
\newcommand{\E}{\mathbb{E}}
\newtheorem{theorem}{Theorem}[section]

\newtheorem{remark}[theorem]{Remark}

\usepackage{xr-hyper}
\usepackage{hyperref}

\hypersetup{
  colorlinks=true,
  linkcolor=blue,
  citecolor=blue,
  filecolor=blue,
  urlcolor=blue,
}

\definecolor{codegreen}{rgb}{0,0.6,0}
\lstdefinestyle{code}{
	commentstyle=\color{codegreen}
}
\begin{document}

\title{\bf Generalised Exponential Kernels for Nonparametric Density Estimation}
\author{Laura M. Craig\footnote{\texttt{email}: \href{mailto:laura.craig@ucdconnect.ie}{laura.craig@ucdconnect.ie}}\,\,\, and Wagner Barreto-Souza\footnote{\texttt{email}: \href{mailto:wagner.barreto-souza@ucd.ie}{wagner.barreto-souza@ucd.ie}}\\
\small \it School of Mathematics and Statistics, University College Dublin, Belfield, Republic of Ireland}

\maketitle

\begin{abstract}
This paper introduces a novel kernel density estimator (KDE) based on the generalised exponential (GE) distribution, designed specifically for positive continuous data. The proposed GE KDE offers a mathematically tractable form that avoids the use of special functions, for instance, distinguishing it from the widely used gamma KDE, which relies on the gamma function. Despite its simpler form, the GE KDE maintains similar flexibility and shape characteristics, aligning with distributions such as the gamma, which are known for their effectiveness in modelling positive data. We derive the asymptotic bias and variance of the proposed kernel density estimator, and formally demonstrate the order of magnitude of the remaining terms in these expressions. We also propose a second GE KDE, for which we are able to show that it achieves the optimal mean integrated squared error, something that is difficult to establish for the former. Through numerical experiments involving simulated and real data sets, we show that GE KDEs can be an important alternative and competitive to existing KDEs.
\end{abstract}

\section{Introduction}
The classic kernel density estimator (KDE) used in nonparametric statistics was introduced by \cite{rosenblatt} and \cite{parzen} to estimate density function with support on the real line $\mathbb R$. Let $X_1,\ldots,X_n$ be an iid (independent and identically distributed) sample from a distribution with density function $f:\mathbb R\rightarrow\mathbb R^+\equiv(0,\infty)$, then the KDE of $f$ assumes the form
\begin{eqnarray*}
\widehat f(x)=\dfrac{1}{n}\sum_{i=1}^n \dfrac{1}{h}K\left(\dfrac{x-X_i}{h}\right),\quad x\in\mathbb R,
\end{eqnarray*}
where $K(\cdot)$ is a density (kernel) function symmetric around zero (for instance, a standard normal one), and $h$ is the bandwidth parameter. Under the conditions that the density $f$ is three times differentiable with bounded third derivatives, and $K$ has finite third moment, we have that the bias and variance of $\widehat f$ are given respectively by
\begin{eqnarray*}
\mbox{bias}\left(\widehat f(x)\right)=\dfrac{1}{2}h^2f''(x)\int_{-\infty}^\infty y^2K(y)dy+o(h^2),
\end{eqnarray*}
as $h\rightarrow0$, and 
\begin{eqnarray*}
\mbox{Var}\left(\widehat f(x)\right)=\dfrac{1}{nh}f(x)\int_{-\infty}^\infty K^2(y)dy+o\left(\dfrac{1}{nh}\right),
\end{eqnarray*}
as $nh\rightarrow\infty$; for instance, see Theorem 6.4.3 from \cite{lehmann}. Since the bias and variance formulas for symmetric kernel density estimators (KDEs) with support on $\mathbb R$ hold under fairly general assumptions, different kernels tend to perform similarly. As a result, bandwidth selection is typically more critical than the choice of kernel.

KDEs defined on the entire real line $\mathbb R$ suffer from an issue known as ``boundary bias" when applied to nonnegative data. This arises because the estimator assigns a nonzero probability to negative values, even though the true data are constrained to be nonnegative. As a result, density estimates near the boundary (i.e., close to zero) become less accurate and more biased \citep{chen}.

Several studies have proposed the use of asymmetric kernels as a modification of the traditional Parzen–Rosenblatt kernel density estimator to better accommodate nonnegative continuous data.
The general form of an asymmetric KDE is
\begin{align*}
	\widehat f (x) = \frac{1}{n} \sum_{i=1}^{n} K^F_{\left(a(x,b),c(b)\right)}(X_{i}),
\end{align*}
where \(X_{1}, \hdots, X_{n}\) are independent and identically distributed (iid) random variables with true density function $f$ with support on $\mathbb R^+$, $K^F_{\left(a(x,b),c(b)\right)}(\cdot)$ is an asymmetric kernel from a distribution $F$ parameters with $a(x,b)$ and $c(b)$ as a function of $(y,b)$ and $b$, respectively, and $b$ denotes the smoothing (bandwidth) parameter.

One of the first attempts to introduce an asymmetric KDE is due to \cite{chen} with a gamma kernel given by
 \begin{align*}
 	K^{GA}_{\left(x/b+1,b\right)}(z) = \frac{z^{x/b}\exp\left\{-z/b\right\}}{b^{x/b + 1}\Gamma\left(x/b + 1\right)}, \quad z>0,
 \end{align*}
which will be refereed as Gam1 KDE along with this paper. To remove the dependence of the bias on the first derivative of $f$ in the interior (i.e., away from $x=0$), \cite{chen} also proposed a second gamma kernel given by $K^{GA}_{\left(\rho_b(x),b\right)}(z)$, with $\rho_b(x)=x/b$ for $x\geq2b$ and $\rho_b(x)=\dfrac{1}{4}(x/b)^2+1$ otherwise. This second gamma KDE will be refereed as Gam2. \cite{scaillet} introduced two alternatives to the gamma kernel based on the inverse-Gaussian (IG) and reciprocal inverse-Gaussian (RIG) distributions, with respective kernels  
 \begin{align*}
 	K^{IG}_{\left(x, 1/b\right)}(z) = \frac{1}{\sqrt{2 \pi b z^3}} \exp \left\{- \frac{1}{2bx}\left(\frac{z}{x} - 2 + \frac{x}{z}\right)\right\},\quad z>0,
 \end{align*}
and
 \begin{align*}
 	K^{RIG}_{\left(1/(x-b), 1/b\right)}(z) = \frac{1}{\sqrt{2 \pi b z}} \exp \left\{- \frac{x-b}{2b}\left(\frac{z}{x-b} - 2 + \frac{x-b}{z}\right)\right\},\quad z>0.
 \end{align*}

Other asymmetric KDEs have been proposed based on the lognormal \citep{bs-ln}, Birnbaum-Saunders \citep{bs-ln,maretal2013,kak2021}, inverse gamma \citep{inv-gamma}, generalised gamma \citep{hirsak2015,general-gamma}, beta prime \citep{beta-prime}, and multivariate elliptical-based Birnbaum-Saunders \citep{kak2022} distributions. The estimation of the first-order derivative of density functions with support on $\mathbb R^+$ has been recently addressed by \cite{funhir2024}. Alternative methods have been proposed by \cite{geewan2018} and \cite{gee2021} based on the local-likelihood transformation and Mellin-Meijer KDEs, respectively.

This paper aims to contribute to the growing body of work on asymmetric kernel density estimators (KDEs) for positive continuous data by introducing a novel KDE based on the generalised exponential (GE) distribution. The proposed GE-based kernel offers a mathematically tractable form, free from special functions such as the gamma function that is involved in the widely used gamma KDE. Despite the gamma KDE’s popularity and versatility in handling non-negative data, the GE KDE presents a simpler alternative while retaining similar flexibility. By leveraging the properties of the GE model, which shares the same general shape as the density and hazard functions of the gamma distribution, this new kernel provides an efficient and accessible option for density estimation.  A second GE KDE is also proposed, for which we are able to show that it achieves the optimal mean integrated squared error, something that is difficult to establish for the former.

The motivation for developing the GE KDEs also lies in the fact that different asymmetric kernels may yield distinct asymptotic properties for bias and variance. This variability highlights the importance of expanding the toolkit of kernels specifically designed for positive data. The GE KDEs provide an appealing alternative for both theoretical analysis and practical implementation. This paper explores their properties, compares their performance to existing methods, and argues for their utility as strong competitors to existing asymmetric kernels. We derive the asymptotic bias and variance of the proposed kernel density estimators, and formally demonstrate the order of magnitude of the remaining terms in these expressions, which is not always addressed in existing papers on the topic. Moreover, through numerical experiments involving simulated and real data sets, we show that GE KDEs can be an important alternative and competitive to existing KDEs.

The paper is organised as follows. Section \ref{sec:GE} introduces the generalised exponential kernel and provides some theoretical results. Section \ref{sec:MC} focuses on Monte Carlo experiments to compare the finite sample performance of the generalised exponential KDE and its existing competitors. A second GE KDE is proposed and explored in Section \ref{sec:GE2}. Section \ref{sec:app} applies the generalised exponential KDEs to two real data sets and compares them to existing KDEs such as the gamma ones by \cite{chen}. Section \ref{sec:conclusion} offers some conclusions on the topic.

\section{Generalised exponential kernels}\label{sec:GE}

The generalised exponential distribution and its properties have been studied by \cite{ge}. A generalised exponential (GE) random variable \(X\) has a density function of the form
\begin{align}
	\label{eq:ge_pdf}
	K^{\mbox{GE}}_{(\alpha, 1/\lambda)}(z) = \alpha \lambda \left(1 - \exp\left\{-\lambda z\right\}\right)^{\alpha - 1} \exp\left\{-\lambda z\right\},\ z > 0,
\end{align}
for \(\alpha, \lambda > 0\). The particular case of interest here for the first GE KDE is when \(\alpha > 1\), as the density function will be unimodal. In this case, the mode is at \(\lambda^{-1}\log\alpha\). 
For \(X \sim \mbox{GE}(\alpha, \lambda\)), the expected value and variance of \(X\) are
\begin{align*}
	\E\left(X\right) = \frac{1}{\lambda} \left[\psi(\alpha + 1) - \psi(1)\right]\quad\mbox{and}\quad \mbox{Var}\left(X\right) = \frac{1}{\lambda^2} \left[\psi '(1) - \psi '(\alpha + 1)\right],
\end{align*} 
where $\psi(z)=d\log\Gamma(x)/dx$ and $\psi'(z)=d^2\log\Gamma(x)/dx^2$ are the digamma and trigamma functions, respectively. The moment generating function of $X$, say $\Phi(t)=E(e^{tX})$, assumes the form
\begin{eqnarray}\label{mgf}
\Phi(t)=\dfrac{\Gamma(\alpha+1)\Gamma(1-t/\lambda)}{\Gamma(\alpha-t/\lambda+1)},
\end{eqnarray}
for $t<\lambda$; see Eq. (5) from \cite{ge}.

Let $X_1,\ldots,X_n$ be an iid sequence of continuous positive random variables with density function $f(x)$. We propose our KDE in terms of the $K^{GE}_{(e^{x/b},b)}(x)$ density function, which has a mode at $x$ with bandwidth parameter $b=1/\lambda$. The first generalised exponential KDE then takes the form
\begin{align}
	\label{eq:ge_final}
	\widehat f_{GE}(x) = \mfrac{1}{n} \sum_{i=1}^{n} K^{GE}_{(e^{x/b},b)}(X_i),\quad x>0.
\end{align}

In what follows, we obtain the asymptotic behaviour of the bias and variance of the proposed GE KDE for the interior and boundary $x$ cases, that are respectively $x/b\rightarrow\infty$ and $x/b\rightarrow c$ as $b\rightarrow0$, where $c$ is a positive real constant.

\begin{theorem}\label{thm:bias} Assume that $f(\cdot)$ is a three-times differentiable function with bounded third derivative. Then, 
\begin{eqnarray*}
	\mbox{bias}\left(\widehat f_{GE}(x)\right) &\equiv& \E\left(\widehat f_{GE}(x)\right) - f(x)\nonumber\\ 
&=&\begin{cases}
b\gamma f'(x) +\frac{1}{2}(\gamma^2+\pi^2/6)b^2f''(x)+o(b^2),\quad\mbox{for} \quad x/b\rightarrow\infty,\\
b[\psi(e^c+1)+\gamma-c]f'(0)+o(b),\quad\mbox{for} \quad x/b\rightarrow c,
\end{cases}    
\end{eqnarray*} 
as $b\rightarrow0$, where $\gamma\approx0.577216$ is the Euler's constant.
\end{theorem}
\begin{proof}
	We have that
	\begin{align}\label{expected_value}
		\E\left(\widehat f_{GE}(x)\right) = \int_{0}^{\infty}  K^{GE}_{\left(e^{x/b}, b\right)} (z) f(z) dz 
		= \E\left(f(\zeta_{x})\right),
	\end{align}
	where \(\zeta_{x} \sim\mbox{GE}(e^{x/b}, b)\). Let us initially consider the interior $x$ case.
    
	We now expand $f(z)$ in Taylor's series around $\mu_x\equiv \E(\zeta_{x})$ (for instance, see Thm 2.5.1 from \cite{lehmann}):
	\begin{align}
		\label{eq:taylor_1}
		f(z) = f(\mu_{x}) + (z - \mu_{x})f'(\mu_{x}) + \mfrac{1}{2} (z - \mu_{x})^2f''(\mu_x)+\mfrac{1}{6} (z - \mu_{x})^3f'''(\omega_{z}),
	\end{align}
	where $\omega_z$ lies between $\mu_x$ and $z$. By replacing (\ref{eq:taylor_1}) in (\ref{expected_value}), we obtain that the bias can be expressed by
\begin{eqnarray}\label{bias_proof}
\mbox{bias}\left(\widehat f_{GE}(x)\right)=f(\mu_{x})-f(x) + \mfrac{1}{2}\mbox{Var}(\zeta_x)f''(\mu_x) +
\mfrac{1}{6}E\left[(\zeta_x-\mu_x)^3f'''(\omega_{\zeta_x})\right],
\end{eqnarray}
with $\mu_{x} = \E(\zeta_{x}) = b \left[\psi\left(e^{x/b} + 1\right) - \psi(1)\right]$.

We use that $\psi(z)=\log z+O(z^{-1})$ as $z\rightarrow\infty$ (for instance, see \cite{polygamma}) and $(e^{x/b}+1)^{-1}=o(b^k)$ as $b\rightarrow0$ for any $k\geq0$ to obtain that 
\begin{eqnarray}\label{aux_bias_proof}
\mu_x&=&b\log(e^{x/b}+1)+o(b^2)-b\psi(1)=x+b\log(1+e^{-x/b})+b\gamma\nonumber\\
&=&x+b\gamma+o(b^2),
\end{eqnarray}
where we used that $\psi(1)=-\gamma$, with $\gamma$ denoting the Euler's constant, and that $$\lim_{b\rightarrow0}\log(1+e^{-x/b})/b=-x\lim_{b\rightarrow0}\frac{1/b^2}{e^{x/b}+1}=0,$$ by using L'Hopital rule, and therefore $b\log(1+e^{-x/b})=o(b^2)$. Hence, by using (\ref{aux_bias_proof}), it follows that
\begin{eqnarray}\label{first_term_bias}
f(\mu_x)-f(x)=b\gamma f'(x)+\frac{1}{2}b^2\gamma^2f''(x)+o(b^2). 
\end{eqnarray}

Another term appearing at the bias expression is $\mbox{Var}(\zeta_x)=b^2[\psi'(1)-\psi'(e^{x/b}+1)]=b^2[\pi^2/6-\psi'(e^{x/b}+1)]$. We now use that $\psi'(z)=O(z^{-1})$ as $z\rightarrow\infty$ (for instance, see \cite{handbook_ctd_frac}) to obtain that $\mbox{Var}(\zeta_x)=b^2\pi^2/6+o(b^2)$. Moreover, $f''(\mu_x)=f''(x)+O(b)$. Therefore,
\begin{eqnarray}\label{midterm}
\mbox{Var}(\zeta_x)f''(\mu_x)=(b^2\pi^2/6+o(b^2))(f''(x)+O(b))=b^2\frac{\pi^2}{6}f''(x)+o(b^2).
\end{eqnarray}

We will now argue that the last term in (\ref{bias_proof}) is $o(b^2)$. The assumption that the third derivative of $f(\cdot)$ is bounded, that is, there is $M>0$ such that $|f'''(z)|\leq M$ for all $z$, is in force. It follows that
\begin{eqnarray*}
|E\left[(\zeta_x-\mu_x)^3f'''(\omega_{\zeta_x})\right]|&\leq& E\left[|\zeta_x-\mu_x|^3|f'''(\omega_{\zeta_x})|\right]\nonumber\\
&\leq& \sqrt{E\left[(\zeta_x-\mu_x)^6\right]}\sqrt{E\left[f'''(\omega_{\zeta_x})^2\right]}\nonumber\\
&\leq&M\sqrt{E\left[(\zeta_x-\mu_x)^6\right]},
\end{eqnarray*}
with Cauchy–Schwarz inequality being used to obtain the second inequality.

We now use a result provided by \cite{kendall} (page 63) that expresses the sixth central moments of a random variable, say $\mu_6$, in terms of cumulants, say $\kappa_j's$, as follows:
\begin{eqnarray}\label{sixthmoment}
\mu_6=\kappa_6+15\kappa_4\kappa_2+10\kappa_3^2+15\kappa_2^3.
\end{eqnarray}

The cumulants of a random variable are obtained by evaluating the derivatives of the cumulant generating function (in short, cgf, which is given by the logarithm of the moment generating function) at zero. By using Eq. (\ref{mgf}), we obtain that the  cgf of $\zeta_x$ is given by
\begin{eqnarray*}
C(t)\equiv \log E(e^{t\zeta_x})=\log\Gamma(1-t b)-\log\Gamma(e^{x/b}-t b+1)+\log\Gamma(e^{x/b}+1),\quad t<1/b.
\end{eqnarray*}

The $j$-th cumulant of $\zeta_x$ ($\kappa_j$) is given by
\begin{eqnarray}\label{cumulants}
\kappa_j=\dfrac{d^jC(t)}{d t^j}\bigg|_{t=0}=(-1)^jb^j\left\{\psi^{(j-1)}(1)-\psi^{(j-1)}(e^{x/b}+1)\right\},\quad j\geq1,
\end{eqnarray}
where $\psi^{(j)}(z)=d^{j-1}\psi(z)/dz^{j-1}$ is the polygamma function. By using the fact that $\psi^{(k)}(z)=O(z^{-k})$ for $k\geq1$, and the expression for cumulants (\ref{cumulants}) in (\ref{sixthmoment}), we immediately obtain that 
$E\left[(\zeta_x-\mu_x)^6\right]=O(b^6)$, and therefore
\begin{eqnarray}\label{remaining}
E\left[(\zeta_x-\mu_x)^3f'''(\omega_{\zeta_x})\right]=o(b^2).
\end{eqnarray}

\noindent Now, we use (\ref{first_term_bias}), (\ref{midterm}), and (\ref{remaining}) in (\ref{bias_proof}) to obtain the desirable result for interior $x$. 

For boundary $x$ ($x/b\rightarrow c$ as $b\rightarrow0$), the bias expression given in (\ref{bias_proof}) still holds.
We have that
$$f(\mu_x)=f\left(b \left[\psi\left(e^{x/b} + 1\right) - \psi(1)\right]\right)=f(0)+ b \left[\psi\left(e^{x/b} + 1\right) - \psi(1)\right] f'(0)+o(b)$$
and
$$f(x)=f(0)+xf'(0)+o(b).$$

Hence,
\begin{eqnarray}\label{firstterm_bound}
f(\mu_x)-f(x)&=&b[\psi(e^{x/b}+1)+\gamma]f'(0)-xf'(0)+o(b)\nonumber\\
&=&b[\psi(e^{c}+1)+\gamma-c]f'(0)+o(b),
\end{eqnarray}
where we have used that $x/b\rightarrow c$ to obtain the second equality. Further, 
\begin{eqnarray}\label{secterm_boundary}
\mbox{Var}(\zeta_x)f''(\mu_x)=\{b^2[\pi^2/6-\psi'(e^c+1)]+o(b^2)\}\{f''(x)+O(b)\}=o(b).
\end{eqnarray}

Finally, the last term of the bias expression is $o(b^2)$ exactly as argued for the interior $x$ case. In other words, (\ref{remaining}) still holds. 
By using (\ref{firstterm_bound}), (\ref{secterm_boundary}), and (\ref{remaining}) in (\ref{bias_proof}), we obtain the asymptotic result for the boundary $x$ case, and this completes the proof of the theorem.
\end{proof}

\begin{theorem}\label{thm:var1} Under the condition of Theorem \ref{thm:bias}, for $nb\rightarrow\infty$ as $b\rightarrow0$, the variance of the GE KDE is 
\begin{eqnarray*}
\mbox{Var}\left(\widehat f_{GE}(x)\right) = 
\begin{cases}
\dfrac{1}{4bn} f(x) + o \left(\dfrac{1}{nb}\right),\quad\mbox{for} \quad x/b\rightarrow\infty,\\
\dfrac{1}{4bn}\dfrac{e^c}{e^c-1/2}f(x) + o \left(\dfrac{1}{nb}\right),\quad\mbox{for} \quad x/b\rightarrow c.
\end{cases}  
\end{eqnarray*}
\end{theorem}

\begin{proof}
The variance of the generalised exponential KDE is 
\begin{eqnarray*}
\mbox{Var}\left(\widehat f_{GE}(x)\right)=\dfrac{1}{n}E\left[\left(K^{GE}_{(e^{x/b},b)}(X)\right)^2\right]+\dfrac{1}{n}\left[E\left(K^{GE}_{(e^{x/b},b)}(X)\right)\right]^2, 
\end{eqnarray*}
where $X$ denotes a random variable with density function $f(\cdot)$. From Theorem \ref{thm:bias}, we have that 
$E\left(K^{GE}_{(e^{x/b},b)}(X)\right)=f(x)+O(b)$. Therefore, 
\begin{eqnarray}\label{eq:var}
\mbox{Var}\left(\widehat f_{GE}(x)\right)=\dfrac{1}{n}E\left[\left(K^{GE}_{(e^{x/b},b)}(X)\right)^2\right]+O(n^{-1}). 
\end{eqnarray}

The term involved in the expectation on the right side of (\ref{eq:var}) can be expressed in terms of the beta exponential (BE) distribution introduced by \cite{betaexp}. A random variable $Z$ follows a BE distribution with shape parameters $a,b>0$ and scale parameter $\lambda>0$ if its density function assumes the form
$g^{BE}_{(a,b,\lambda^{-1})}(z)=\dfrac{\lambda}{B(a,b)}e^{-b\lambda z}(1-e^{-bz})^{a-1}$, for $z>0$, where $B(a,b)=\dfrac{\Gamma(a)\Gamma(b)}{\Gamma(a+b)}$ is the beta function. We denote $Z\sim\mbox{BE}(a,b,\lambda^{-1})$. It follows that
\begin{eqnarray}\label{auxvar1}
E\left[\left(K^{GE}_{(e^{x/b},b)}(X)\right)^2\right]&=&\dfrac{e^{2x/b}}{b}B(2e^{x/b}-1,2)E\left[f(\eta_x)\right]\nonumber\\
&=&\dfrac{1}{2b}\dfrac{e^{x/b}}{2e^{x/b}-1}E\left[f(\eta_x)\right], 
\end{eqnarray}
where $\eta_x\sim\mbox{BE}(2e^{x/b}-1,2,b)$. The mean and the variance of $\eta_x$ are given by $\mu_x^*\equiv E(\eta_x)=b[\psi(2e^{x/b}+1)-\psi(2)]$ and $\mbox{Var}(\eta_x)=b^2[\psi'(2)-\psi'(2e^{x/b}+1)]$. We now expand $f(\cdot)$ in Taylor's series around $\mu_x^*$ and take the expectation (in a similar fashion as done in Theorem \ref{thm:bias}) to obtain that
\begin{eqnarray}\label{auxvar2}
E\left[f(\eta_x)\right]=f(\mu_x^*)+\dfrac{1}{2}\mbox{Var}(\eta_x)f''(\mu_x^*)+\frac{1}{6}E\left[(\zeta_x-\mu^*_x)^3f'''(\omega^*_{\eta_x})\right],
\end{eqnarray}
where $\omega_{\eta_x}^*$ lies between $\mu_x^*$ and $\eta_x$. 

We have that $\mu_x^*\rightarrow x$ as $b\rightarrow0$. Moreover, $\mbox{Var}(\eta_x)=b^2[\psi'(2)-\psi'(2e^{x/b}+1)]=o(b)$ and $f''(\mu_x)=f''(x)+O(b)$ (using similar arguments as those from proof of Theorem \ref{thm:bias}) for both interior and boundary $x$. So, $\mbox{Var}(\eta_x)f''(\mu_x)=o(b)$. We now claim that $E\left[(\zeta_x-\mu^*_x)^3f'''(\omega^*_{\eta_x})\right]=o(b^2)$. The way to justify this result follows the same steps as those from Theorem \ref{thm:bias}, and using the fact that the cumulant generating function of $\eta_x$ is given by
$\Phi_{\eta_x}(t)\equiv \log E(e^{t\eta_x})=\log\Gamma(2-bt)-\log\Gamma(2e^{x/b}-bt+1)+\log\Gamma(2e^{x/b}+1)$,
for $t<2/b$, where we have used Eq. (3.1) from \cite{betaexp} to get the mgf of a beta exponential distribution. The cumulants of $\eta_x$ are given by $\kappa_j=(-b)^j\{\psi^{(j-1)}(2)-\psi^{(j-1)}(2e^{x/b}+1)\}=O(b^j)$, which holds for both interior and boundary $x$.

Using the above results, (\ref{auxvar1}), and (\ref{auxvar2}) in (\ref{eq:var}), we obtain that 
\begin{eqnarray}\label{varfinal}
\mbox{Var}\left(\widehat f_{GE}(x)\right)&=&\dfrac{1}{2nb}\dfrac{e^{x/b}}{2e^{x/b}-1}\left\{f(\mu_x^*)+o(b)+o(b^2)\right\}+O(n^{-1})\nonumber\\
&=&\dfrac{1}{2nb}\dfrac{e^{x/b}}{2e^{x/b}-1}f(\mu_x^*)+o\left(\dfrac{1}{bn}\right)+O(n^{-1})\nonumber\\ 
&=&\dfrac{1}{4nb}\dfrac{e^{x/b}}{2e^{x/b}-1}\left[f(x)+O(b)\right]+o\left(\dfrac{1}{bn}\right)\nonumber\\ 
&=&\dfrac{1}{4nb}\dfrac{e^{x/b}}{e^{x/b}-1/2}f(x)+o\left(\dfrac{1}{bn}\right),
\end{eqnarray}
where we used that $f(\mu_x^*)=f(x)+O(b)$ for both interior and boundary $x$ and that $O(n^{-1})=o(1/(nb))$. From (\ref{varfinal}), the stated expressions for the variance under the interior and boundary cases are immediately obtained. 
\end{proof}

With the asymptotic mean and variance of the GE kernel density estimators, we can obtain an expression for the mean squared error (MSE), which is 
\begin{eqnarray*}
	\mbox{MSE}\left(\widehat f_{GE}(x)\right) &=& \left(b\gamma f'(x) + \dfrac{1}{2}\left(\gamma^2 + \dfrac{\pi^2}{6}\right)b^2f''(x) + o(b^2)\right)^2 + \dfrac{1}{4bn}f(x) + o\left(\dfrac{1}{nb}\right)\nonumber \\
    &=&b^3\gamma\left(\gamma^2 + \dfrac{\pi^2}{6}\right) f'(x)f''(x)+b^2\gamma^2 f'(x)^2+ \dfrac{1}{4bn}f(x) +  o(b^3 + (nb)^{-1}), 
\end{eqnarray*}
while the mean integrated squared error (MISE) is approximately
\begin{align*}
	\mbox{MISE}\left(\widehat f_{GE}(x)\right) \approx b^3\gamma\left(\gamma^2 + \dfrac{\pi^2}{6}\right)\int_0^\infty f'(x)f''(x)dx+b^2\gamma^2\int_0^\infty f'(x)^2dx + \dfrac{1}{4bn}, 
\end{align*}
where we have assumed that $\displaystyle\int_0^\infty f''(x)^2dx<\infty$ and $\displaystyle\int_0^\infty f'(x)^2dx<\infty$ (which implies that $\displaystyle\int_0^\infty f'(x)f''(x)dx$ is integrable by using Cauchy-Schwarz inequality).

We cannot obtain an explicit form for the optimal bandwidth $b$ that minimises the approximated MISE, but this can be done numerically. Also, the integrals involved in that expression can be computed by replacing the first two derivatives with some estimators or by assuming some true density function. In this paper, the rule-of-thumb introduced by \cite{silverman} will be adopted when calculating the bandwidth parameter.

\section{Monte Carlo experiments}\label{sec:MC}
To investigate the finite sample properties of our KDE and its competitors, simulations were carried out to compare the MISEs of the GE kernel, the two gamma kernels \citep{chen} and the RIG kernel \citep{scaillet} (the IG case was not presented since such a kernel did not produce competitive results). Five different underlying distributions are considered as shown in Table ~\ref{tab:distributions}, involving gamma, inverse gamma, inverse Weibull, mixture of gamma, mixture of inverse gamma, and mixture of inverse Weibull. Two sample sizes, $n=100$ and $n=500$, were considered, and the number of Monte Carlo replications was 1000. Each kernel is tested on all six of the configurations, and the MISE results are recorded in boxplots as shown in Figures \ref{fig:bias_plots1} and \ref{fig:bias_plots2}. Bandwidth was computed using Silverman's rule-of-thumb \citep{silverman}. 
\begin{table}[ht!]
	\begin{center}
		\begin{tabular}{ccc}
			\hline
			Configuration & Distribution  \\ 
			\hline
			A & Gamma(25,0.5) \\
			B & Inverse Gamma(25,150) \\
			C & Inverse Weibull(5,800) \\
			D & Gamma Mixture((2/3,1/3), (\(\Gamma(25,0.5)\), \(\Gamma(5,2)\))) \\
			E & Inverse Gamma Mixture((2/3,1/3), (\(IGam(25,150)\), \(IGam(30,5)\)))\\
        	\hline
		\end{tabular}
		\caption{Distributions/configurations used in the Monte Carlo simulation experiments.}
		\label{tab:distributions}
	\end{center} 
\end{table}
\begin{figure}[ht!]
	\begin{subfigure}{8cm}
		\includegraphics[width=8cm]{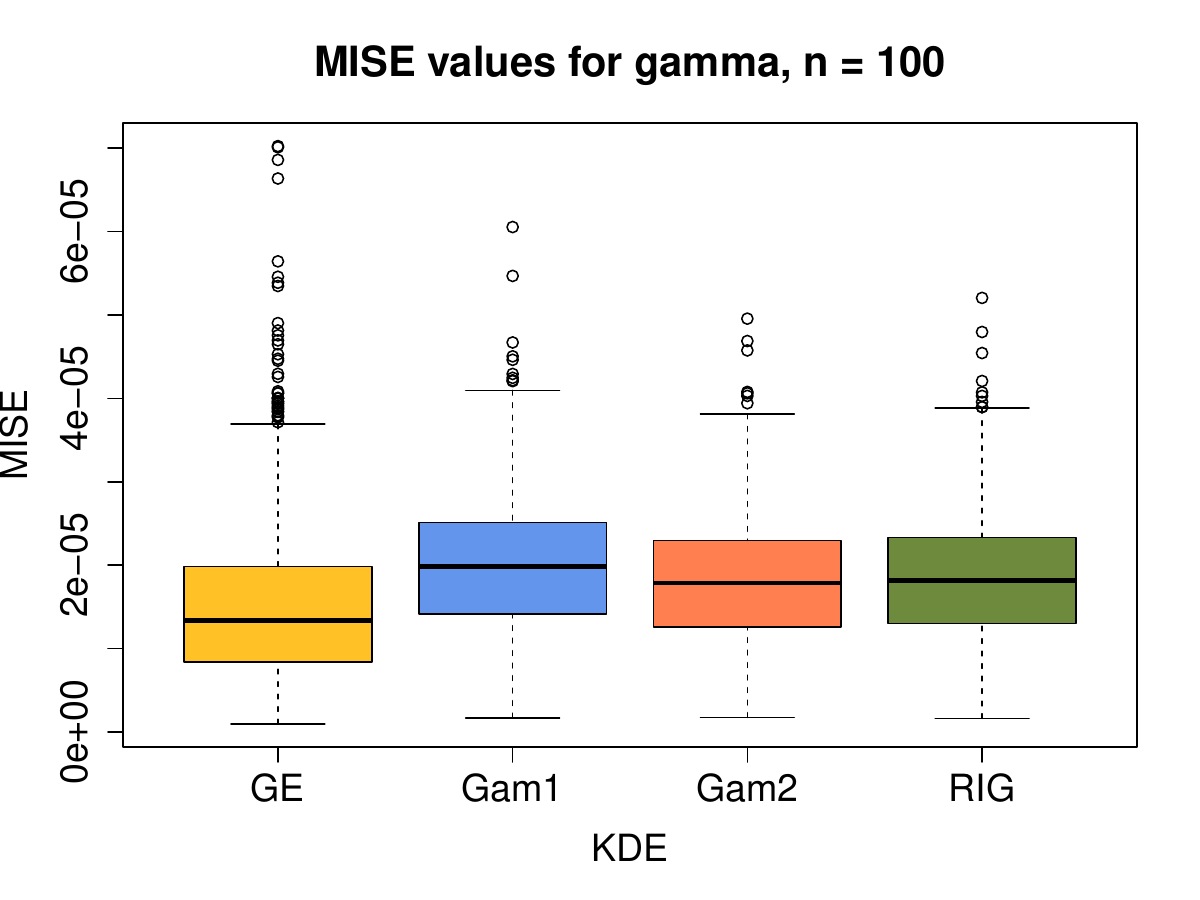}
	\end{subfigure}
	\begin{subfigure}{8cm}
		\includegraphics[width=8cm]{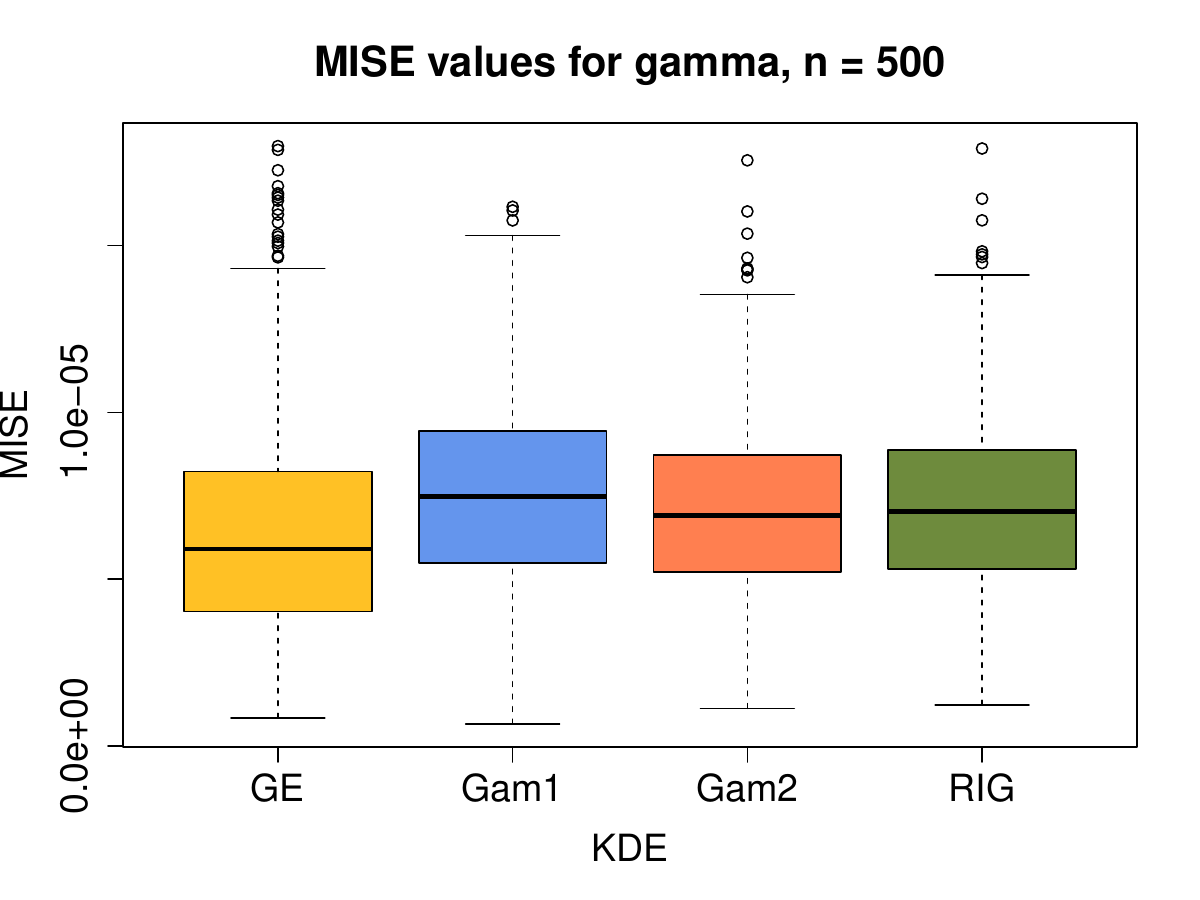} 	
	\end{subfigure}
	\begin{subfigure}{8cm}
		\includegraphics[width=8cm]{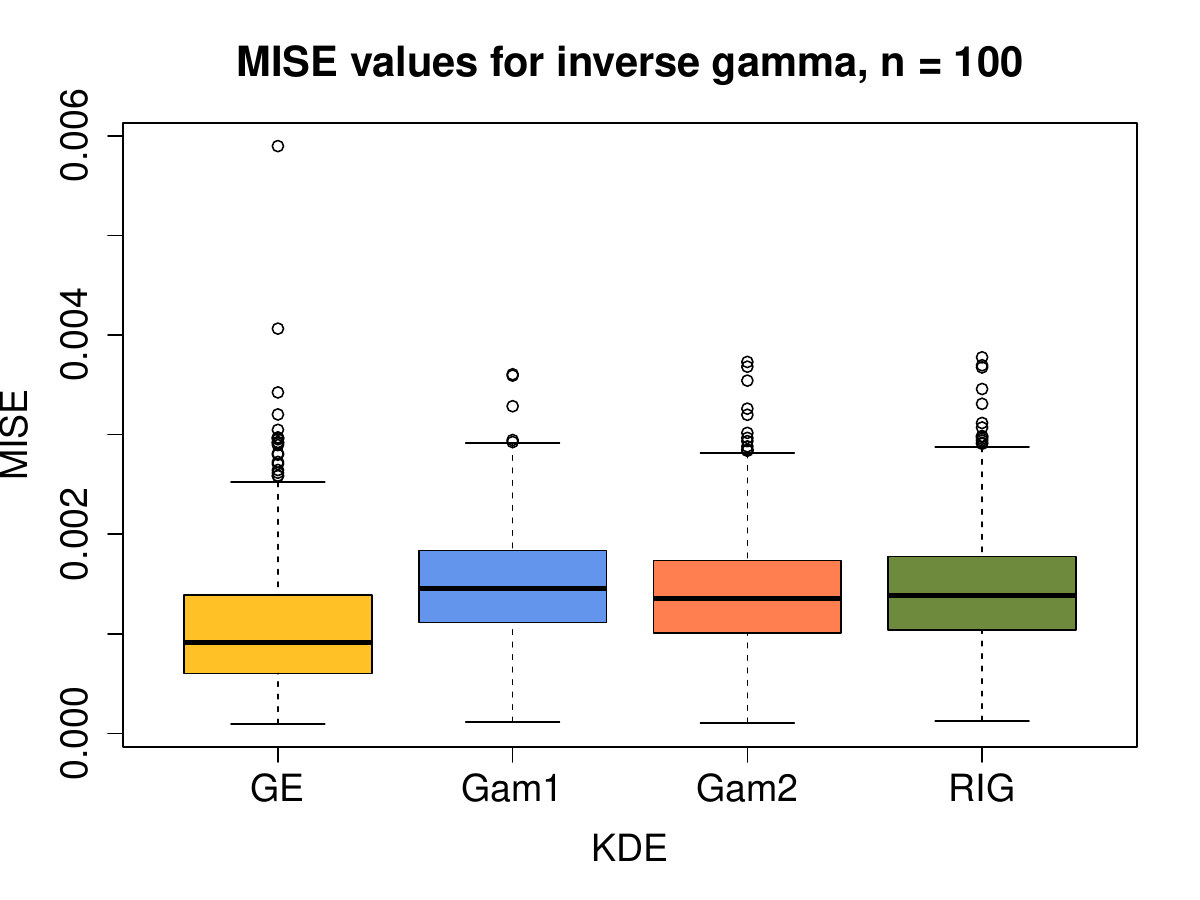}
	\end{subfigure}
	\begin{subfigure}{8cm}
		\includegraphics[width=8cm]{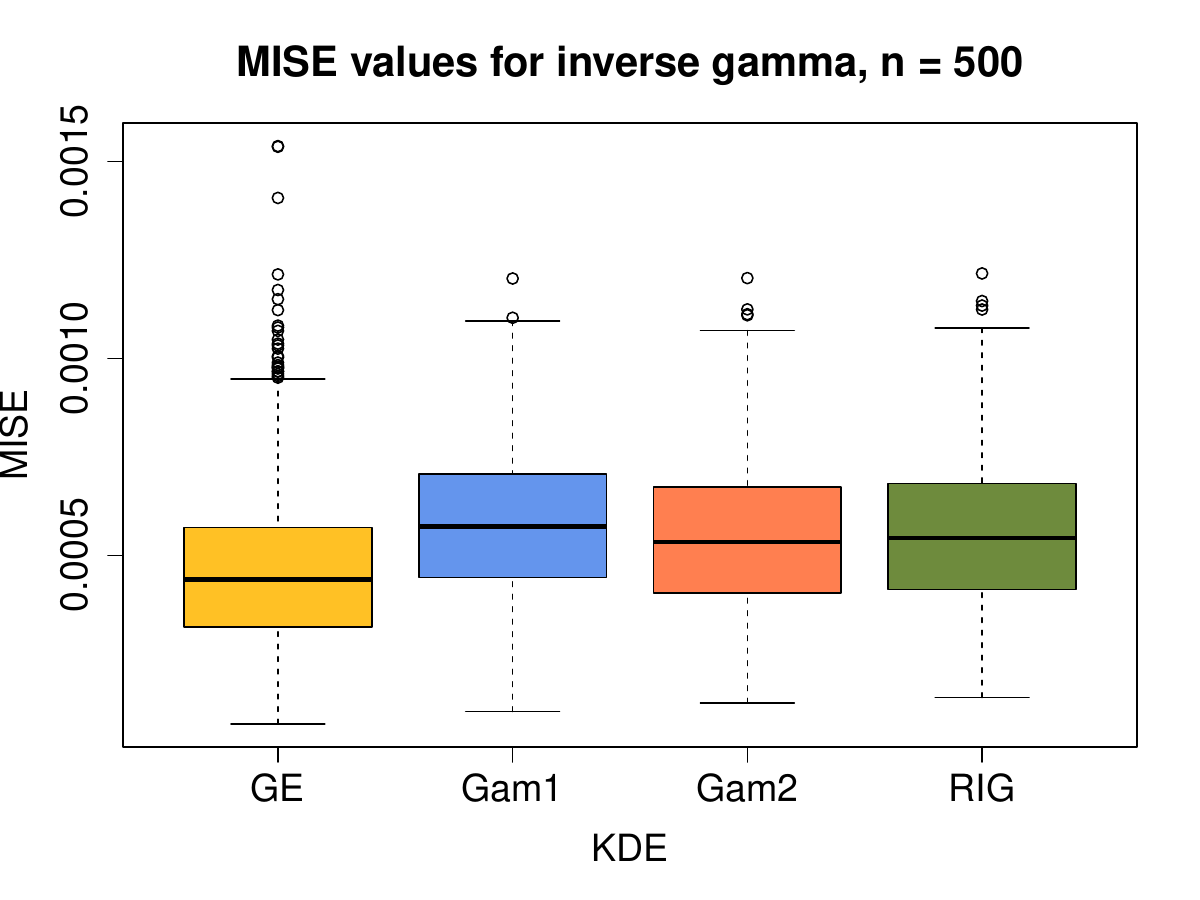}
	\end{subfigure}
	\begin{subfigure}{8cm}
		\includegraphics[width=8cm]{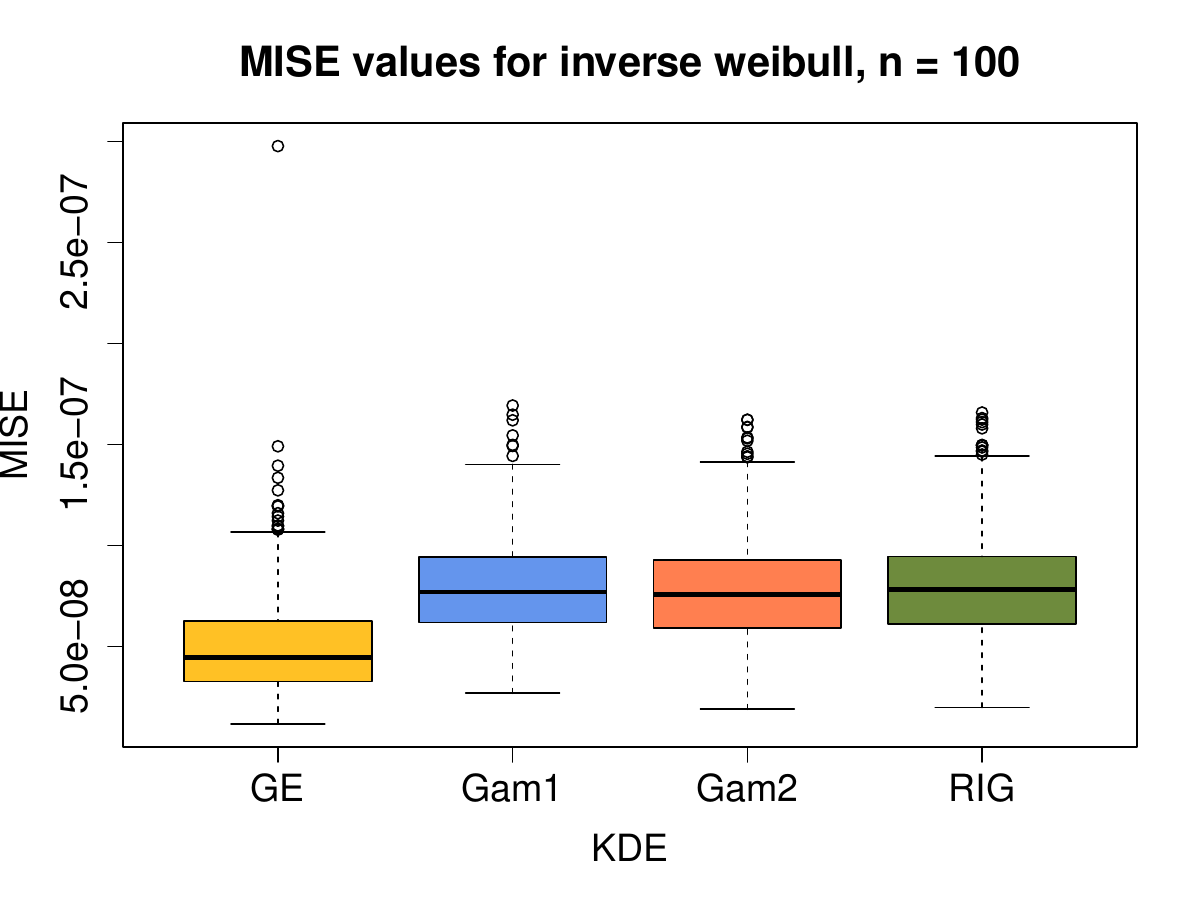}
	\end{subfigure}
	\begin{subfigure}{8cm}
		\includegraphics[width=8cm]{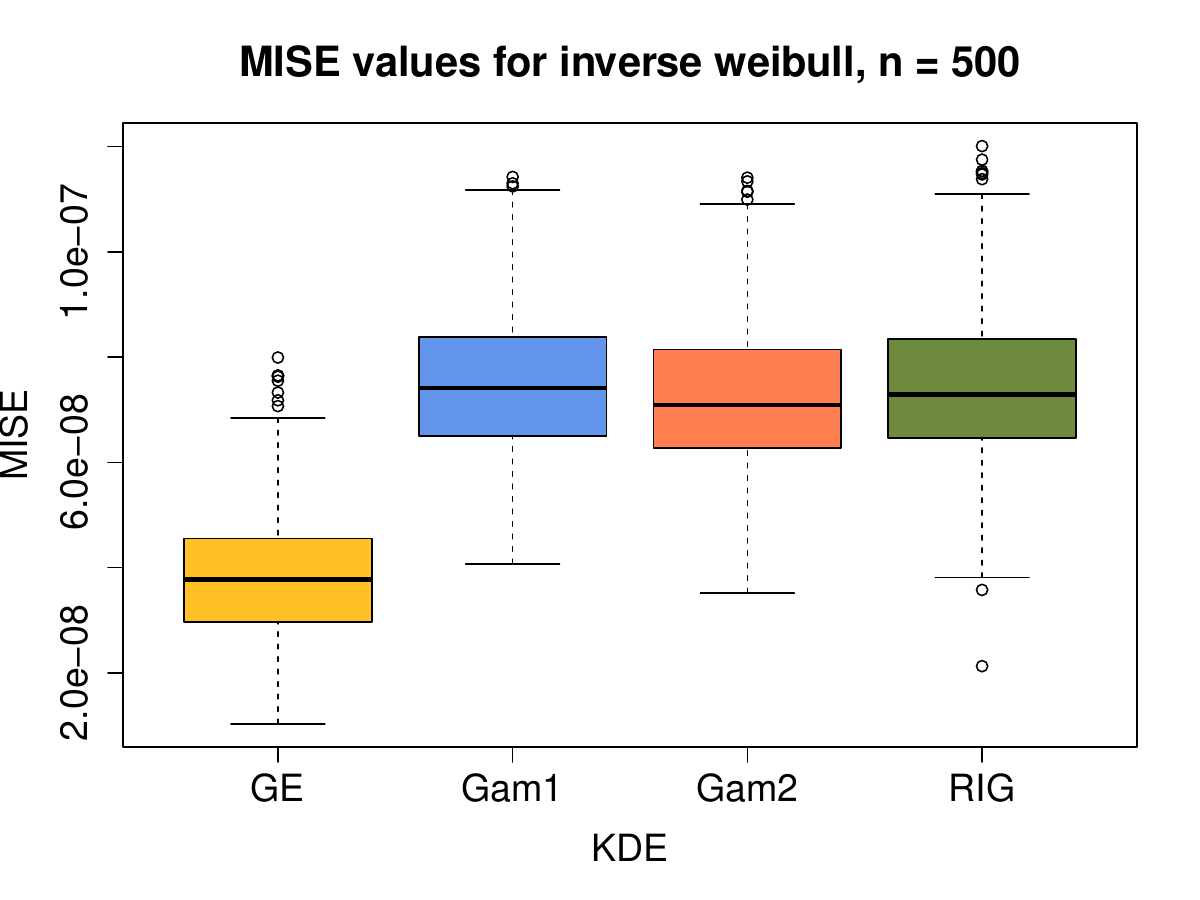}	
	\end{subfigure}
    	\caption{Boxplots of the MISEs based on the Configurations A, B, and C.}
	\label{fig:bias_plots1}
\end{figure}
\begin{figure}[!h]
	\begin{subfigure}{8cm}
		\includegraphics[width=8cm]{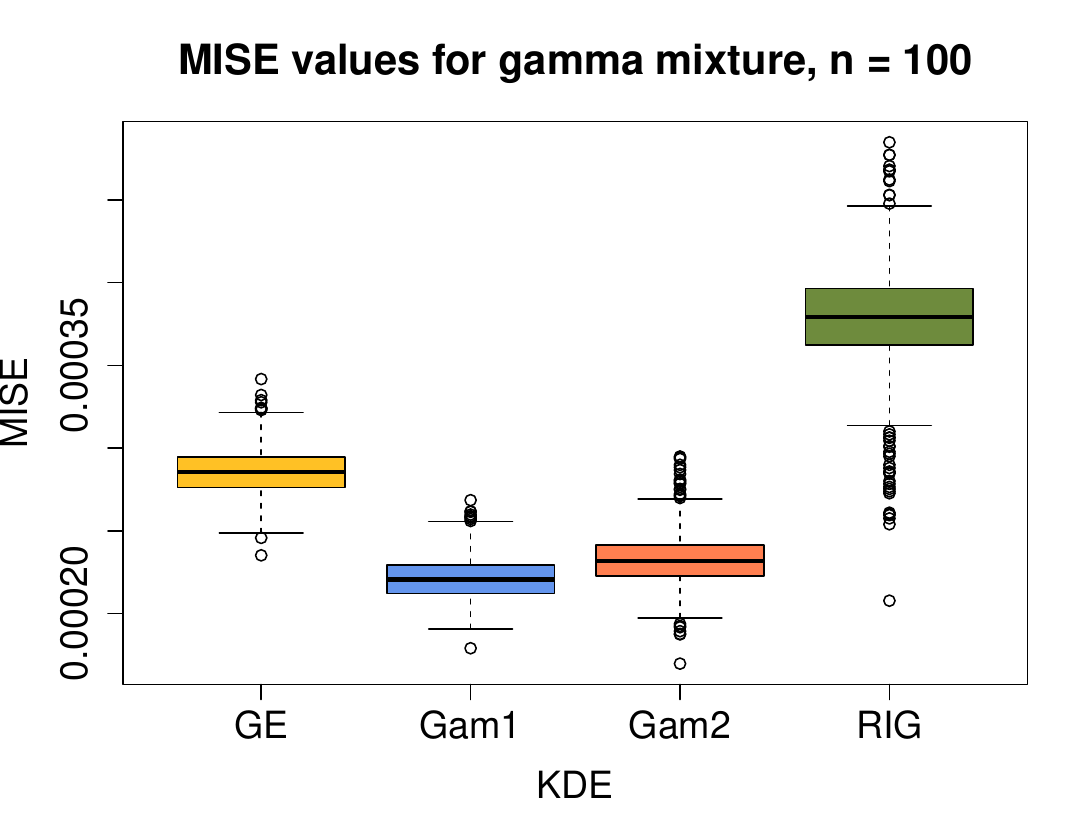}
	\end{subfigure}
	\begin{subfigure}{8cm}
		\includegraphics[width=8cm]{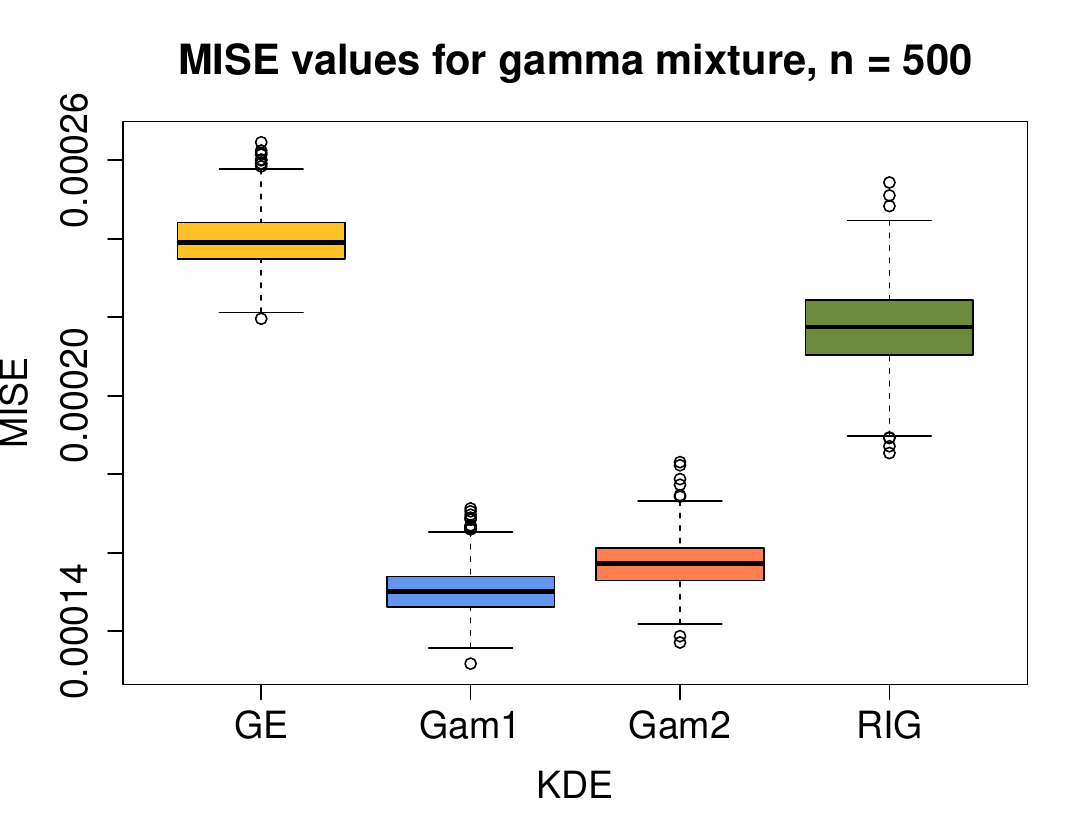}
	\end{subfigure}
	\begin{subfigure}{8cm}
		\includegraphics[width=8cm]{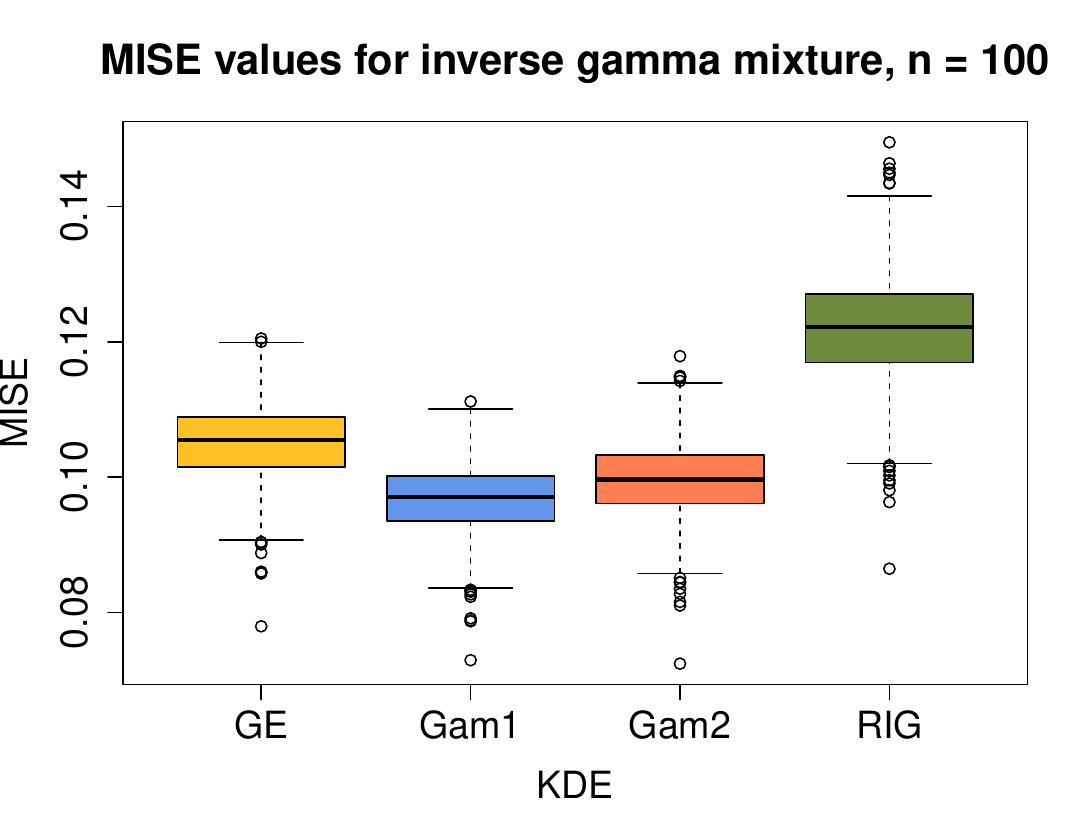}
	\end{subfigure}
	\begin{subfigure}{8cm}
		\includegraphics[width=8cm]{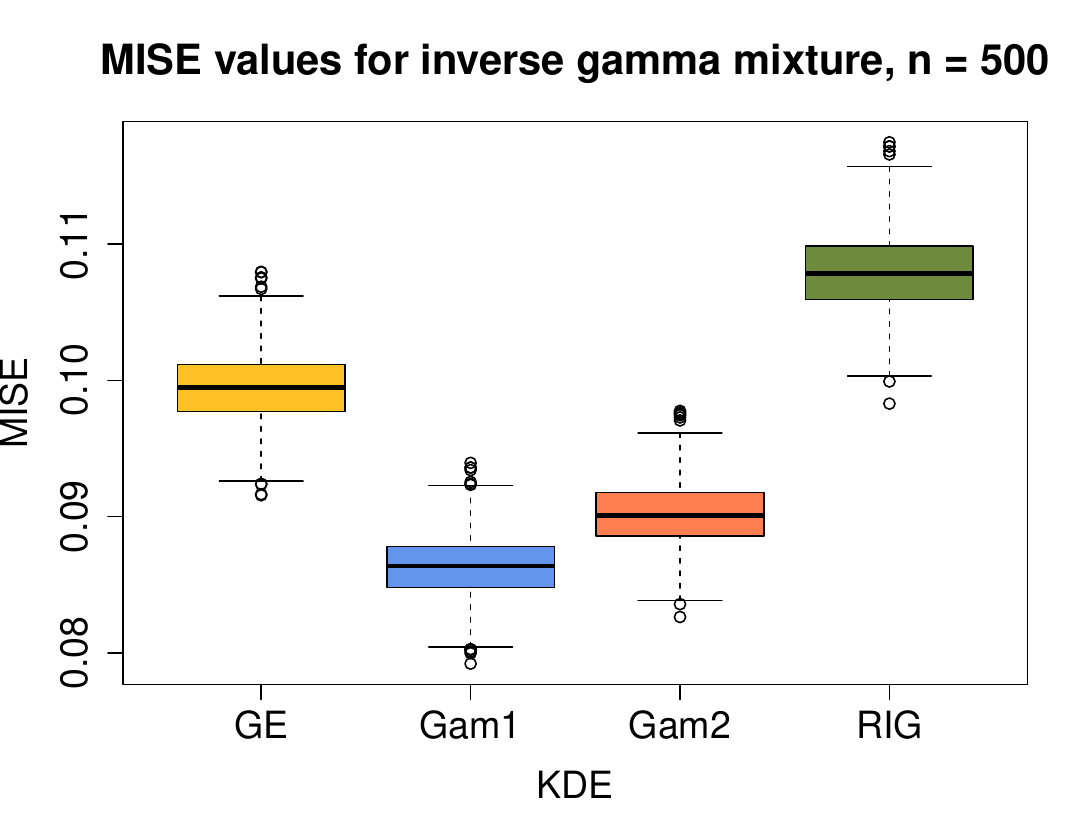}	
	\end{subfigure}
	\caption{Boxplots of the MISEs based on the Configurations D and E.}
	\label{fig:bias_plots2}
\end{figure}

For configurations A, B, and C, the GE MISE values are lower than those of the Gamma and RIG kernels. The variances are similar across all kernels, and the MISE values of the Gamma and RIG kernels are similar. The MISE results for the GE and Gamma kernels are competitive for configurations D and E, sharing a similar variance with Gam1 coming in with the lowest values. The RIG results are the largest in both magnitude and variance.  As can be seen from Table ~\ref{tab:mean_mise}, the average MISE values resulting from use of the GE kernel in the simulation scenarios are more often than not the lowest out of the KDEs presented. 

\begin{table}[ht!]
	\begin{center}
		\begin{tabular}{cccccc}
			\hline
			Configuration & $n$ & GE & Gam1 & Gam2 & RIG  \\ 
			\hline
	    	\multirow{2}{*}{A} & 100 & \num[math-rm=\mathbf]{1.55e-05} & \num{2.01e-05} & \num{1.82e-05} & \num{1.86e-05}  \\
			& 500 & \num[math-rm=\mathbf]{6.40e-06} & \num{7.59e-06} & \num{7.07e-06} & \num{7.19e-06}  \\
			\multirow{2}{*}{B} & 100 & \num[math-rm=\mathbf]{1.05e-03} & \num{1.49e-03} & \num{1.40e-03} & \num{1.43e-03}  \\
			& 500 & \num[math-rm=\mathbf]{4.64e-04} & \num{5.81e-04} & \num{5.46e-04} & \num{5.54e-04}  \\
	    	\multirow{2}{*}{C} & 100 & \num[math-rm=\mathbf]{4.96e-08} & \num{7.91e-08} & \num{7.69e-08} & \num{7.90e-08}  \\
	    	& 500 & \num[math-rm=\mathbf]{3.83e-08} & \num{7.47e-08} & \num{7.21e-08} & \num{7.41e-08}  \\
			\multirow{2}{*}{D} & 100 & \num{3.64e-04} & \num{3.63e-04} & \num[math-rm=\mathbf]{3.61e-04} & \num{3.64e-04} \\
			& 500 & \num{3.42e-04} & \num{3.32e-04} & \num[math-rm=\mathbf]{3.29e-04} & \num{3.31e-04}  \\
			\multirow{2}{*}{E} & 100 & \num[math-rm=\mathbf]{1.05e-01} & \num{9.67e-02} & \num{9.95e-02} & \num{1.20e-01}  \\
			& 500 & \num[math-rm=\mathbf]{9.95e-02} & \num{8.63e-02} & \num{9.02e-02} & \num{1.08e-01} \\
        	\hline
		\end{tabular}
		\caption{Average MISE values from Monte Carlo simulations based on GE, Gam1, Gam2, and RIG KDEs.}
		\label{tab:mean_mise}
	\end{center}
\end{table}

\section{An alternative GE KDE with optimal MISE}\label{sec:GE2}

We here propose a second GE KDE with a reparameterised GE distribution in terms of the mean. As a consequence, we will see that the dependence of the bias on the first-order derivative of $f$ will be eliminated, which is commonly desirable. Furthermore, we will explicitly show that this new KDE achieves the optimal MISE order of $O(n^{-4/5})$. 

The second GE KDE (referred to as GE2 in what follows) is defined by
\begin{align}
	\label{eq:ge2_final}
	\widehat f_{GE2}(x) = \mfrac{1}{n} \sum_{i=1}^{n} K^{GE}_{(\nu({x/b}),b)}(X_i),\quad x>0,
\end{align}
where we have defined $\nu({x/b})=\psi^{-1}\left(x/b+\psi(1)\right)-1$, with $\psi^{-1}(\cdot)$ being the inverse function of $\psi(\cdot)$. In terms of \texttt{R} implementation in our numerical experiments, we consider the function \texttt{idigamma()} from the package \texttt{bzinb} for the computation of the inverse of the digamma function involved in (\ref{eq:ge2_final})). The expected value of a $\mbox{GE}(\nu(x/b),b)$ distribution is $x$. Therefore, GE2 KDE is defined with observations as the mean of the kernels, in contrast with the first GE KDE, where the mode is used. The next result provides the asymptotic mean and variance of the alternative GE kernel density estimator.

\begin{theorem}\label{thm:bias2} Assume that $f(\cdot)$ is a three-times differentiable function with bounded third derivative. Then, 
\begin{eqnarray*}
	\mbox{bias}\left(\widehat f_{GE2}(x)\right)=
b^2\dfrac{\pi^2}{12}f''(x)+o(b^2),   
\end{eqnarray*} 
as $b\rightarrow0$. Moreover, the asymptotic variance of $\widehat f_{GE2}(x)$ assumes the form as that of the estimator $\widehat f_{GE}(x)$ provided in Theorem \ref{thm:var1}.
\end{theorem}

\begin{proof}
By using results from \cite{iss2016} (who formally proved the asymptotic series for the inverse of the digamma function given by Ramanujan), we have that $\psi^{-1}(z)=e^z+o(1)$ as $z\rightarrow\infty$. Hence, it follows that $\nu(x/b)=e^{x/b}e^{-\gamma}+o(1)$, where $\gamma$ is the Euler's constant.
In other words, the term $\nu(x/b)$ involved in GE2 KDE behaves similarly asymptotically (when $b\rightarrow0$) as $e^{x/b}$ used in GE KDE, apart from the constant $e^{-\gamma}$. This implies that the proofs for the asymptotic expressions for the bias and variance for the GE2 KDE follow exactly as those for GE KDE, and therefore, they are omitted.
\end{proof}

The mean squared error (MSE) of the GE2 KDE is 
\begin{eqnarray*}
	\mbox{MSE}\left(\widehat f_{GE2}(x)\right) = b^4\dfrac{\pi^4}{144}f''(x)^2+\dfrac{1}{4bn}f(x)+  o(b^4 + (nb)^{-1}), 
\end{eqnarray*}
while the MISE is approximately
\begin{align}\label{amise}
	\mbox{MISE}\left(\widehat f_{GE2}(x)\right) \approx b^4\dfrac{\pi^4}{144}\int_0^\infty f''(x)^2dx+\dfrac{1}{4bn}, 
\end{align}
where we are assuming that $\displaystyle\int_0^\infty f''(x)^2dx<\infty$. In this case, the optimal bandwidth parameter (say $b^*$) that minimises the approximate MISE can be computed explicitly and is given by
\begin{eqnarray}\label{opt}
b^*=\left(\dfrac{9}{\pi^4 \int_0^\infty f''(x)^2dx}\right)^{1/5}n^{-1/5}.
\end{eqnarray}

Using the optimal bandwidth (\ref{opt}) in (\ref{amise}), we obtain that $\mbox{MISE}\left(\widehat f_{GE2}(x)\right)=O(n^{-4/5})$, which achieves the optimal order of existing KDEs such as Gaussian, gamma, IG, and RIG, among others. We now provide two important remarks.

\begin{remark}
Note that the bandwidth parameter 
$b$ for our GE kernels is comparable in magnitude to the Gaussian bandwidth 
$b$, whereas for the gamma, IG, and RIG kernels, the appropriate comparison is with their squared bandwidth, that is, $b^2$.
\end{remark}

\begin{remark} (Finiteness of integrals) The GE2 KDE only requires $\displaystyle\int_0^\infty f''(x)^2dx<\infty$, whereas GE1 KDE additionally requires $\displaystyle\int_0^\infty f'(x)^2dx<\infty$. To use the gamma KDEs by \cite{chen}, it is necessary that $\displaystyle\int_0^\infty \left(xf''(x)\right)^2dx<\infty$ and $\displaystyle\int_0^\infty f'(x)^2dx<\infty$. On the other hand, the RIG and IG kernels by \cite{scaillet} assume that $\displaystyle\int_0^\infty \left(xf''(x)\right)^2dx<\infty$ and $\displaystyle\int_0^\infty x^{-1/2}f(x)dx<\infty$, and 
$\displaystyle\int_0^\infty \left(x^3f''(x)\right)^2dx<\infty$ and $\displaystyle\int_0^\infty x^{-3/2}f(x)dx<\infty$, respectively. This highlights the importance of considering a range of KDEs when analysing positive data, as different methods rely on different underlying assumptions. Moreover, our proposed GE kernels, particularly GE2, require substantially simpler conditions than those imposed by existing KDEs.
\end{remark}

We now conclude this section by presenting a small Monte Carlo with an
additional Configuration F involving a mixture of inverse Weibull distributions, more specifically: \\Inverse Weibull Mixture((2/3,1/3), (\(IW(5,800)\), \(IW(10,400)\))). Figure \ref{fig:add_conf_F} presents the boxplots of the MISEs based on Configuration F under GE, GE2, Gam1, Gam2, and RIG kernels. The average MISE values are reported in Table \ref{tab:mean_mise_conf_F}. These results show that the GE kernels present the smallest MISE values, with GE2 being the best approach under the configuration considered for both sample sizes $n=100$ and $n=500$.

\begin{figure}[ht!]
	\begin{subfigure}{8cm}
		\includegraphics[width=8cm]{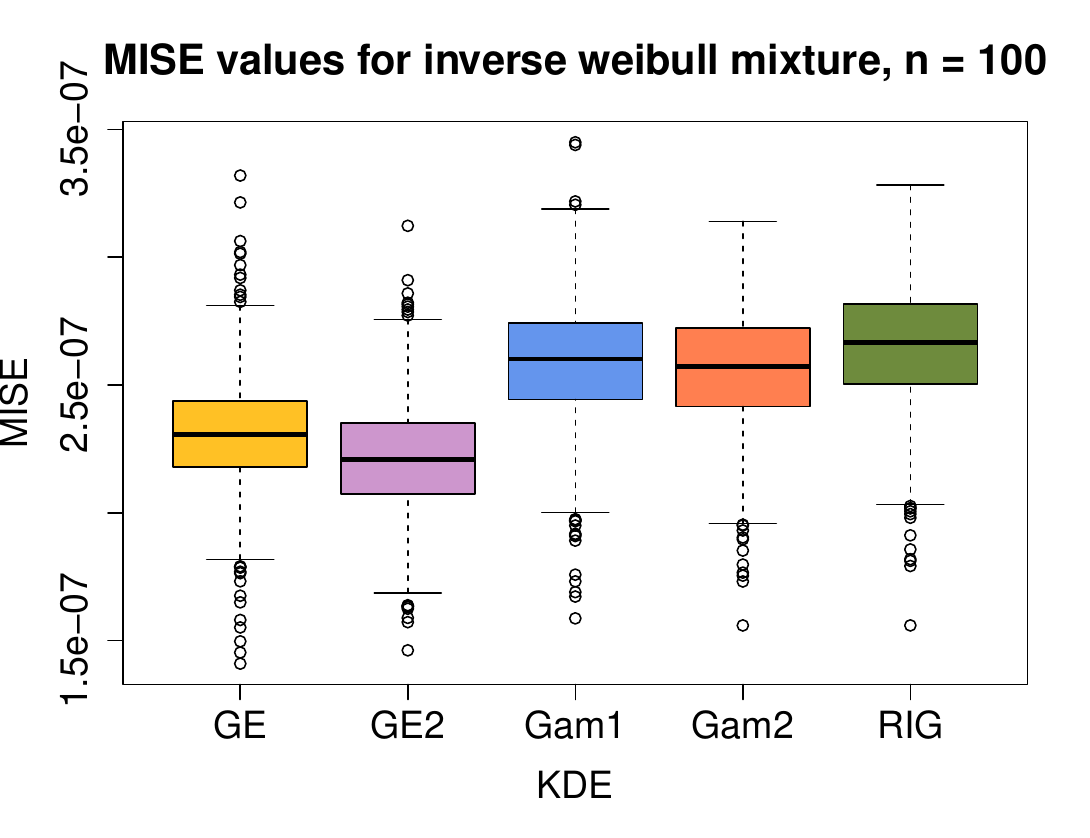}
	\end{subfigure}
	\begin{subfigure}{8cm}
		\includegraphics[width=8cm]{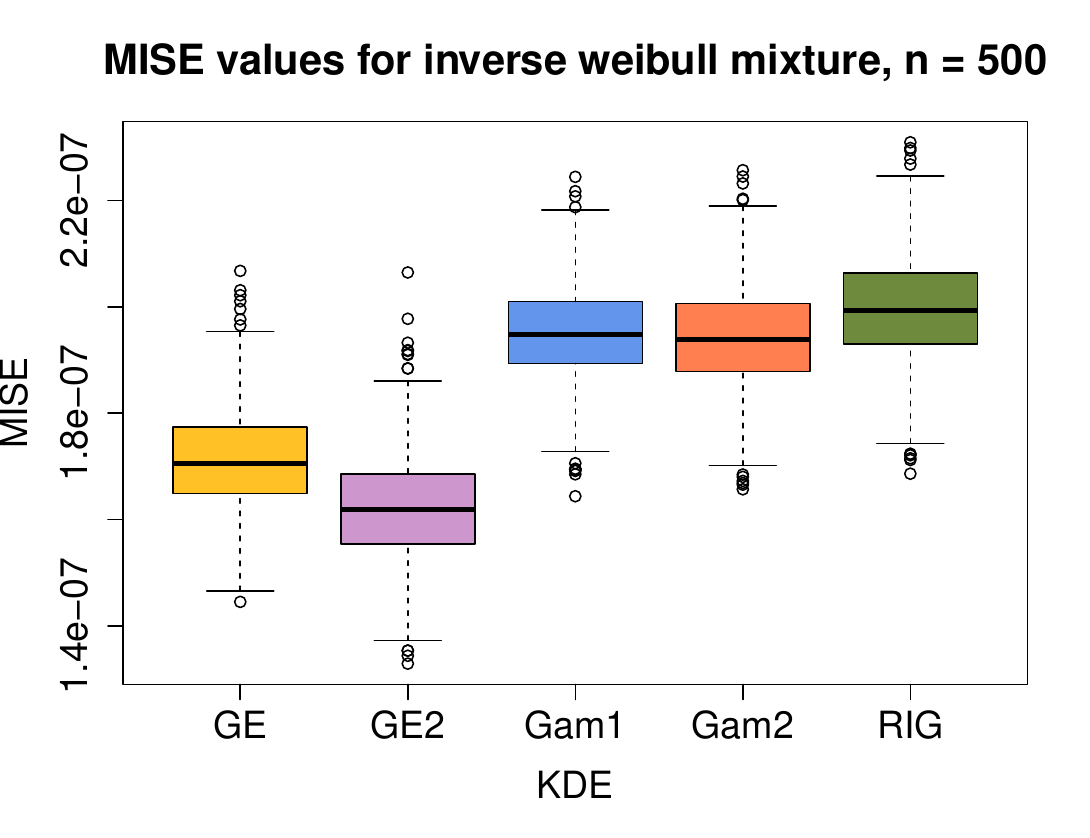} 	
	\end{subfigure}
    \caption{Boxplots of the MISEs based on Configuration F.}\label{fig:add_conf_F}
\end{figure}

\begin{table}[ht!]
	\begin{center}
		\begin{tabular}{ccccccc}
			\hline
			Configuration & $n$ & GE & GE2 & Gam1 & Gam2 & RIG \\ 
			\hline
	    	\multirow{2}{*}{F} & 100 & \num{2.31e-07} & \num[math-rm=\mathbf]{2.22e-07} & \num{2.58e-07} & \num{2.56e-07} & \num{1.22e-01} \\
            & 500 & \num{1.71e-07} & \num[math-rm=\mathbf]{1.62e-07} & \num{1.95e-07} & \num{1.94e-07} & \num{1.08e-01}\\
        	\hline
		\end{tabular}
		\caption{Average MISE values from Monte Carlo simulation for Configuration F under GE, GE2, Gam1, Gam2, and RIG KDEs.}
		\label{tab:mean_mise_conf_F}
	\end{center}
\end{table}

In the next section, we compare the GE kernels with their competitors using two real data applications: one involving the lifetimes (in years) of retired women with temporary disabilities enrolled in the Mexican public insurance system who died in 2004, and another based on a snowfall dataset collected in Grand Rapids in December 1983.

\section{Real data applications}\label{sec:app}

This section explores the application of the GE kernels to two real datasets and compares them with the existing competitors: Gam1, Gam2, and RIG KDEs. The bandwidth estimate in this section is calculated using Silverman's rule-of-thumb as before (with the dependence on the sample size matching the kernels' optimality).  Cross-validation to calculate bandwidth was experimented with; however, the results yielded were poor compared to the Silverman bandwidth choices. Cross-validation bandwidth estimates were particularly bad for the RIG kernel. 

\subsection{Mexican Institute of Social Security data}
This data was considered by \cite{ig-dist} as an application of a class of inverse Gaussian type distributions. The data set contains 280 observations in the form of the lifetime (in years) of retired women with temporary disabilities incorporated into the Mexican public insurance system, who died in 2004. The data was originally provided by the Mexican Institute of Social Security. 

All of the kernels are reasonably smooth, with the GE kernels peaking in the right area. The Gamma and RIG kernels do not peak in quite the right area. 
The approximate MISE were $1.020\times10^{-5}$, $1.015\times10^{-5}$, $3.213\times10^{-5}$, $3.307\times10^{-5}$, and $3.358\times10^{-5}$ for the GE, GE2, Gam1, Gam2, and, RIG kernels, respectively. The smallest MISE is achieved by GE2, closely accompanied by GE. Figure \ref{fig:lifetime_app} displays the KDEs along with the histogram for the Mexican Institute of Social Security data. From the plot, we observe that GE and GE2 better capture the features of the data. GE exhibits slightly more local variation, whereas GE2 is smoother while still peaking at the appropriate points.

\begin{figure}[ht!]
	\includegraphics[width=18cm]{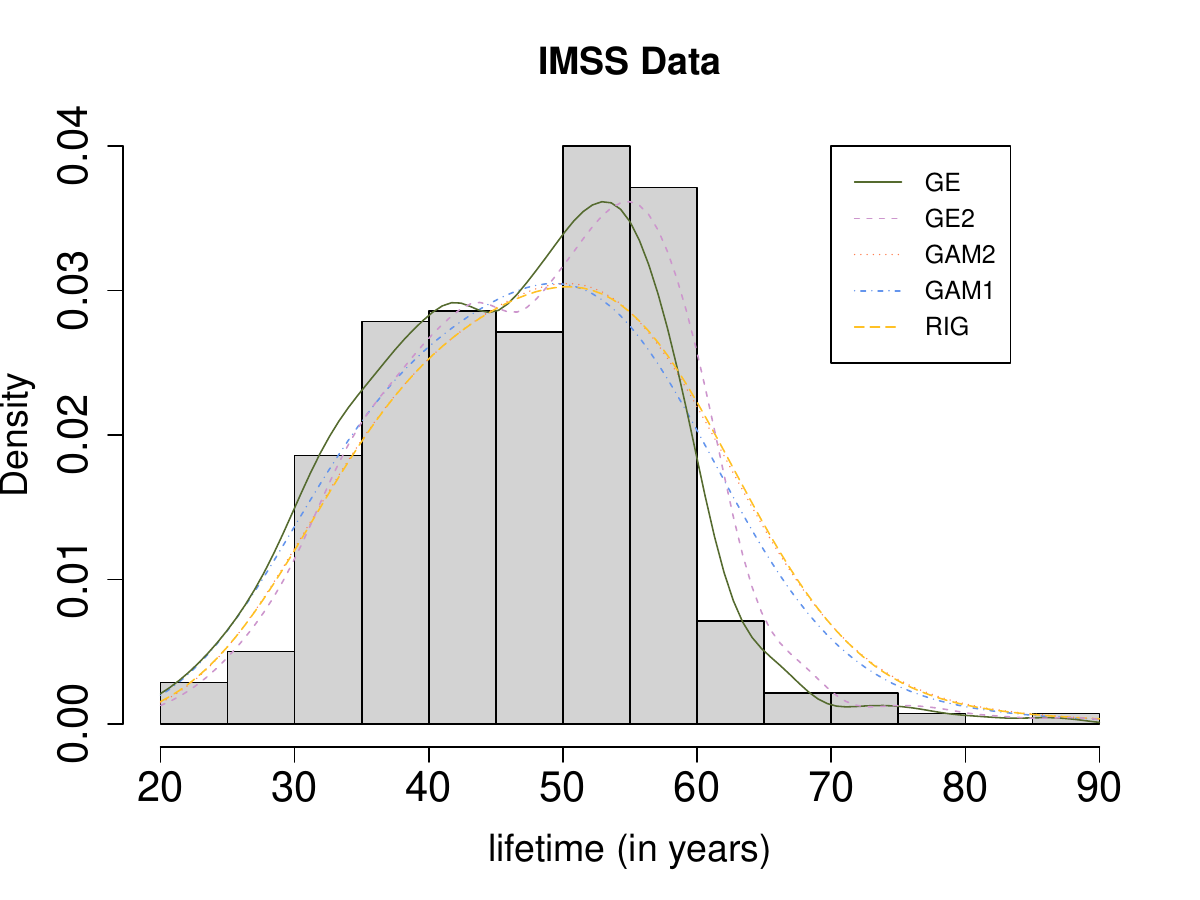} 
	\caption{Kernel density estimates and histogram for the Mexican Institute of Social Security data.}\label{fig:lifetime_app}
\end{figure}

\subsection{Snow data}

This data set was analysed by \cite{beta-prime} to illustrate nonparametric density estimation based on the beta prime kernel. It is available in the dataset \texttt{SnowGR} as part of the \texttt{MosaicData} package. The data consists of 119 observations of snowfall collected in Grand Rapids, MI, in December 1983. The data was originally obtained from NOAA. 

All of the kernels are reasonably smooth and peak around the right area. However, the GE kernel follows the shape of the data a little more closely than the Gamma and RIG kernels.
The approximate MISE for the snow data were $2.630\times10^{-6}$, $4.584\times10^{-6}$, $3.152\times10^{-6}$, $5.000\times10^{-6}$, and $6.097\times10^{-6}$ respectively under the GE, GE2, Gam1, Gam2, and RIG kernels. The smallest MISE was achieved by the first GE kernel, followed by Gam1 kernel. Figure \ref{fig:snow_app} shows the kernel density estimates alongside the histogram for the snowfall data. As with the previous example, we visually observe that GE and GE2 better capture the data's features. GE displays slightly more local variation, while GE2 remains smoother, yet still peaks at the appropriate points.

\begin{figure}[ht!]
	\includegraphics[width=18cm]{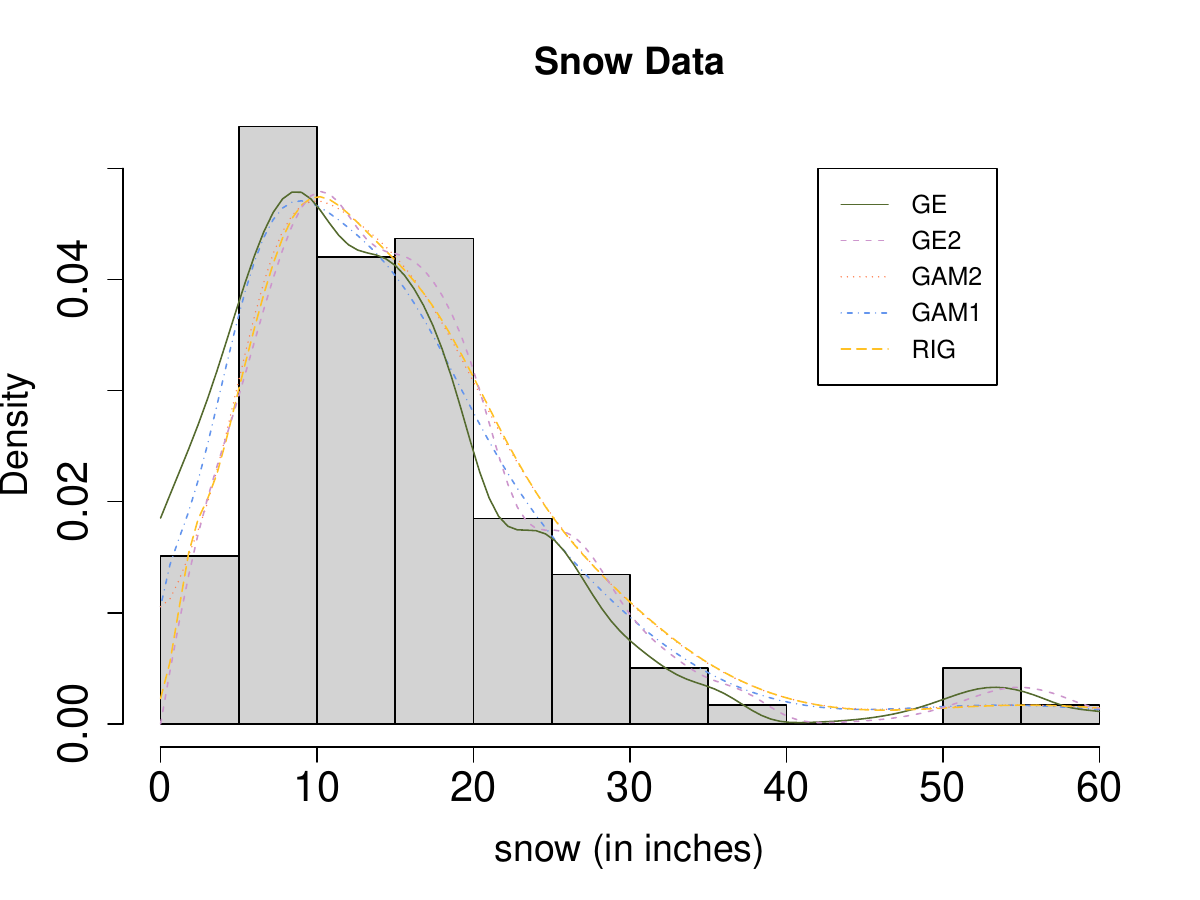} 
	\caption{Kernel density estimates and histogram for the snow data.}
\end{figure}\label{fig:snow_app}

\section{Concluding remarks}\label{sec:conclusion}

This paper introduced a new asymmetric kernel density estimator for positive continuous data based on the generalised exponential (GE) distribution. By exploiting the tractable structure of the GE model, the proposed KDE avoids the use of special functions such as the gamma function required by the gamma KDE, while preserving comparable flexibility in shape and behaviour. As a result, the GE KDE offers a simpler yet effective alternative for density estimation on the positive real line, with advantages for both analytical derivations and practical implementation. Beyond its computational simplicity, the GE KDE enriches the class of asymmetric kernels available for positive data, an important consideration given that different kernels can exhibit distinct asymptotic bias and variance properties. We provide a detailed theoretical analysis of the proposed estimator, deriving explicit expressions for its asymptotic bias and variance and formally establishing the order of the remaining terms. This level of rigor addresses gaps in the existing literature, where such remainder terms are not always fully characterised. A second GE kernel was also introduced, and we show that it achieves the optimal mean integrated squared error, an optimality result that could not be established for the first GE kernel.

The practical relevance of the GE KDEs is further supported by numerical studies using simulated and real data sets, which demonstrate that the proposed estimator performs competitively with well-established asymmetric KDEs. These results underscore the value of expanding the range of kernels tailored to positive data and illustrate that simple distributional choices can yield estimators with strong theoretical and empirical performance.

Overall, the GE KDEs constitute a meaningful addition to the asymmetric KDE framework, offering a balance between mathematical tractability, theoretical robustness, and practical effectiveness. Future research may explore data-driven bandwidth selection, extensions to multivariate settings, and applications in domains where positive-valued data are prevalent.

\section*{Acknowledgements}
This publication has emanated from research conducted with the financial support of Taighde Éireann – Research Ireland under Grant 18/CRT/6049. For the purpose of Open Access, the authors have applied a CC BY public copyright licence to any Author Accepted Manuscript version arising from this submission.

\bibliography{refs.bib}

@article{iss2016,
	author = {Aziz Issaka},
	journal = {The Ramanujan Journal},
	pages = {291–302},
	title = {On {R}amanujan’s inverse digamma approximation},
	volume = {39},
	year = {2016}}

@article{ig-dist,
	author = {Antonio Sanhueza and V{\'\i}ctor Leiva and N. Balakrishnan},
	journal = {Metrika},
	pages = {31-49},
	title = {A new class of inverse Gaussian type distributions},
	volume = {68},
	year = {2008}}

@book{kendall,
	author = {M.G. Kendall},
	edition = {4th},
	publisher = {Griffin},
	title = {Advanced Theory of Statistics: Volume 1},
	year = {1948}}

@article{geewan2018,
	author = {Gery Geenens and Craig Wang},
	journal = {Journal of Computational and Graphical Statistics},
	pages = {822-835},
	title = {Local-likelihood transformation kernel density estimation for positive random variables},
	volume = {27},
	year = {2018}}

@article{gee2021,
	author = {Gery Geenens},
	journal = {Annals of the Institute of Statistical Mathematics},
	pages = {953–977},
	title = {Mellin–{M}eijer kernel density estimation on {$\mathbb R^+$}},
	volume = {73},
	year = {2021}}

@article{funhir2024,
	author = {Benedikt Funke and Masayuki Hirukawa},
	journal = {Journal of Nonparametric Statistics},
	pages = {994-1017},
	title = {Density derivative estimation
using asymmetric kernels},
	volume = {36},
	year = {2024}}

@article{kak2022,
	author = {Yoshihide Kakizawa},
	journal = {Journal of Multivariate Analysis},
	pages = {104834},
	title = {Multivariate elliptical-based {B}irnbaum–{S}aunders kernel density estimation for nonnegative data},
	volume = {187},
	year = {2022}}

@article{hirsak2015,
	author = {Masayuki Hirukawa and Mari Sakudo},
	journal = {Journal of Nonparametric
Statistics},
	pages = {41-63},
	title = {Family of the generalised gamma
kernels: a generator of asymmetric kernels for nonnegative data},
	volume = {27},
	year = {2015}}

@article{maretal2013,
	author = {C Marchant and K Bertin and V Leiva and H Saulo},
	journal = {Computational Statistics and Data Analysis},
	pages = {1-15},
	title = {Generalized {B}irnbaum–{S}aunders kernel density estimators and an analysis of financial data},
	volume = {63},
	year = {2013}}

@article{kak2021,
	author = {Y. Kakizawa},
	journal = {Computational Statistics and Data Analysis},
	pages = {107249},
	title = {A class of {B}irnbaum-{S}auders type kernel density estimators for nonnegative data},
	volume = {161},
	year = {2021}}

@book{lehmann,
	author = {E.L. Lehmann},
	edition = {1st},
	publisher = {Springer},
	title = {Elements of Large-Sample Theory},
	year = {1999}}

@book{silverman,
	author = {Bernard W. Silverman},
	edition = {1st},
	publisher = {Chapman and Hall},
	title = {Density Estimation for Statistics and Data Analysis},
	year = {1986}}

@article{bs-ln,
	author = {Xiaodong Jin and Janusz Kawczak},
	journal = {Annals of Economics and Finance},
	pages = {103-124},
	title = {Birnbaum-{S}aunders and Lognormal Kernel Estimators for Modelling Durations in High Frequency Financial Data},
	volume = {4},
	year = {2003}}

@article{inv-gamma,
	author = {Yoshihide Kakizawa and Gaku Igarashi},
	journal = {Journal of the Korean Statistical Society},
	pages = {194-207},
	title = {Inverse gamma kernel density estimation for nonnegative data},
	volume = {46},
	year = {2017}}

@article{beta-prime,
	author = {Elif Er{\c c}elik and Mustafa Nadar},
	journal = {Communications in Statistics - Theory and Methods},
	pages = {325-342},
	title = {Nonparametric density estimation based on beta prime kernel},
	volume = {49},
	year = {2020}}

@article{general-gamma,
	author = {Gaku Igarashi and Yoshihide Kakizawa},
	journal = {Journal of Nonparametric Statistics},
	pages = {598-639},
	title = {Generalised gamma kernel density estimation for nonnegative data and its bias reduction},
	volume = {30},
	year = {2018}}

@article{chen,
	author = {Song Xi Chen},
	journal = {Annals of the Institute of Statistical Mathematics},
	pages = {471-480},
	title = {Probability Density Function Estimation using Gamma Kernels},
	volume = {52},
	year = {2000}}

@article{scaillet,
	author = {O. Scaillet},
	journal = {Journal of Nonparametric Statistics},
	pages = {217-266},
	title = {Density Estimation using Inverse and Reciprocal Inverse Gaussian Kernels},
	volume = {16},
	year = {2004}}

@article{rosenblatt,
	author = {Murray Rosenblatt},
	journal = {The Annals of Mathematical Statistics},
	pages = {832-837},
	title = {Remarks on Some Nonparametric Estimates of a Density Function},
	volume = {27},
	year = {1956}}

@article{parzen,
	author = {Emanuel Parzen},
	journal = {The Annals of Mathematical Statistics},
	pages = {1065-1076},
	title = {On estimation of a probability density function and mode},
	volume = {33},
	year = {1962}}

@article{betaexp,
	author = {Saralees Nadarajah and Samuel Kotz},
	journal = {Reliability Engineering and System Safety},
	pages = {689-697},
	title = {The beta exponential distribution},
	volume = {91},
	year = {2006}}

@book{handbook_ctd_frac,
	author = {Annie Cuyt and Vigdis Brevik Petersen and Brigitte Verdonk and Haakon Waadeland and William B. Jones},
	publisher = {Springer},
	title = {Handbook of Continued Fractions for Special Functions},
	year = {2008}}

@book{polygamma,
	author = {H. M. Srivastava and Junesang Choi},
	publisher = {Elsevier},
	title = {Zeta and q-Zeta Functions and Associated Series and Integrals},
	year = {2012}}

@article{ge,
	author = {Rameshwar D. Gupta and Debasis Kundu},
	journal = {Journal of Statistical Planning and Inference},
	month = {March},
	pages = {3537-3547},
	title = {Generalized exponential distribution: Existing results and some recent developments},
	volume = {137},
	year = {2007}}

\end{document}